%% file: fvst.tex
\documentclass[11pt]{article}

\usepackage{amsmath,amsthm}
\newtheorem{definition}{Definition}
\newtheorem{theorem}{Theorem}

\newtheorem{lemma}{Lemma}

\usepackage{amssymb,calc}
\usepackage{mdframed}
\usepackage{microtype}

\usepackage{vmargin}
\setmarginsrb{1in}{1in}{1in}{1in}{0mm}{0mm}{0mm}{7mm}

\begin{document}

\title{Faster Exact and Parameterized Algorithm for Feedback Vertex Set in Tournaments}
 \author{Mithilesh Kumar\thanks{Department of Informatics, University of Bergen, Norway \texttt{\{mithilesh.kumar|daniello\}@ii.uib.no}} 
\addtocounter{footnote}{-1}
 \and Daniel Lokshtanov
}

\maketitle


\begin{abstract}
\input{abstract}

\end{abstract}
\section{Introduction}
\input{intro}

\section{Preliminaries}\label{pre}
\input{prelim}

\section{Finding an Undeletable, Evenly Spread Out Set}\label{family}
\input{familyUndel}

\section{\textsc{Faster Algorithm for Tournament Feedback Vertex Set}}\label{parafeed}
\input{mainAlgorithm}


\section{Balanced Edge Partition Theorem}\label{par}
\input{partition}
\section{$d$-\textsc{Feedback Vertex Cover} with Undirected Degree at Most One}\label{dfeed}
\input{fvcDivideConquer}

\appendix


\nocite{DBLP:journals/ipl/GutinJ14}
\bibliography{fvst}{}



\end{document}

%% file: abstract.tex
A {\em tournament} is a directed graph $T$ such that every pair of vertices are connected by an arc. A {\em feedback vertex set} is a set $S$ of vertices in $T$ such that $T - S$ is acyclic. In this article we consider the {\sc Feedback Vertex Set} problem in tournaments. Here input is a tournament $T$ and integer $k$, and the task is to determine whether $T$ has a feedback vertex set of size at most $k$. We give a new algorithm for {\sc Feedback Vertex Set in Tournaments}. The running time of our algorithm is upper bounded by $O(1.618^k + n^{O(1)})$ and by $O(1.46^n)$. Thus our algorithm simultaneously improves over the fastest known parameterized algorithm for the problem by Dom et al.  running in time $O(2^kk^{O(1)} + n^{O(1)})$, and the fastest known exact exponential time algorithm by Gaspers and Mnich with running time $O(1.674^n)$. 
On the way to prove our main result we prove a new partitioning theorem for undirected graphs. In particular we show that the vertices of any undirected $m$-edge graph of maximum degree $d$ can be colored white or black in such a way that for each of the two colors, the number of edges with both endpoints of that color is between $m/4-d/2$ and $m/4$+$d/2$.

%% file: intro.tex
A {\em feedback vertex set} in a graph $G$ is a vertx set $S$ such that $G - S$ is acyclic. For undirected graphs this means that $G - S$ is a forest, while for directed graphs this implies that $G - S$ is a directed acyclic graph (DAG). 
In the {\sc Feedback Vertex Set} (FVS) problem we are given as input an {\em undirected} graph $G$ and integer $k$, and asked whether there exists a feedback vertex set of size at most $k$. The corresponding problem for directed graphs is called {\sc Directed Feedback Vertex Set} (DFVS). Both problems are NP-complete~\cite{GJ79} and have been extensively studied from the perspective of approximation algorithms~\cite{BafnaBF99,EvenNSS98}, parameterized algorithms~\cite{ChenLLOR08,CyganNPPRW11,KociumakaP14}, exact exponential time algrithms~\cite{Razgon07,XiaoN15} as well as graph theory~\cite{erdHos1965independent,reed1996packing}.

In this paper we consider a restriction of DFVS, namely the {\sc Feedback Vertex Set in Tournaments} (\textsf{TFVS}) problem, from the perspective of parameterized algorithms and exact exponential time algorithms. We refer to the textbooks of Cygan et al.~\cite{pc_book} and Fomin and Kratsch~\cite{FominK10} for an introduction to these fields. A {\em tournament} is a dircted graph $T$ such that every pair of vertices is connected by an arc, and FVST is simply DFVS when the input graph is required to be a tournament.

Even this restricted variant of DFVS has applications in voting systems and rank aggregation~\cite{Dom201076}, and is quite well-studied~\cite{CaiDZ00,Dom201076,GaspersM13,RamanS06}. \textsf{TFVS} was shown to be fixed parameter tractable by Raman and Saurabh~\cite{RamanS06}, who obtained an algorithm with running time $O(2.42^k\cdot n^{O(1)})$. In 2006, Dom et al.~\cite{DomGHNT06conf} (see also~\cite{Dom201076}) gave 
an algorithm for \textsf{TFVS} with running time $2^kn^{O(1)}$. Prior to our work this was the fastest known parameterized algorithm for the problem. The fastest exact exponential time algorithm for \textsf{TFVS} was due to Gaspers and Mnich~\cite{GaspersM13} and has running time $O(1.674^n)$.

Our main result is a new algorithm for \textsf{TFVS}. The running time of our algorithm is upper bounded by $O(1.618^k + n^{O(1)})$ and by $O(1.46^n)$. Thus, we give a single algorithm that simultaneously significantly improves over the previously best known parameterized algorithm and exact exponential time algorithm for the problem. It is worth noting that the algorithm of Gaspers and Mnich~\cite{GaspersM13} also lists all inclusion minimal feedback vertex sets in the input tournament, while our algorithm can not be used for this purpose.

On the way to proving our main result we prove a new balanced edge partition theorem for general undirected graphs. In particular we show that the vertices of any undirected $m$-edge graph $G$ of maximum degree $d$ can be colored white or black in such a way that for each of the two colors, the number of edges with both endpoints of that color is between $m/4-d/2$ and $m/4$+$d/2$. This partition theorem is of independent interest, and we believe it will find further applications both in algorithms and in graph theory.

\smallskip
\noindent
{\bf Methods.} As a preliminary step our algorithm applies the kernel of Dom et al.~\cite{Dom201076} to ensure that the number of vertices in the input tournament is upper bounded by $O(k^3)$. After this step, our algorithm has three phases.


In the first phase the algorithm finds, in subexponential time, a ``large enough'' set $M$ of vertices {\em disjoint} from the solution $S$ sought for, such that $M$ is evenly distributed in the topological ordering of $T - S$. 
From the set $M$ we can infer a rough sketch of the unique topological ordering of $T - S$ without knowing the solution $S$. More concretely, every vertex $v$ gets a tentative position in the ordering, and we know that if $v$ is not deleted, then $v$'s position in the topological order of $T-S$ is close to this tentative position. 

We can now use this tentative ordering to identify {\em conflicts} between two vertices $u$ and $v$. Two vertices $u$ and $v$ are in conflict if their tentative positions are so far apart that we know the order in which they have to appear in the topological sort of $T-S$, but the arc between $u$ and $v$ goes in the opposite direction. Thus, if $u$ and $v$ are in conflict then at least one of them has to be in the solution feedback vertex set $S$.

The second phase of the algorithm eliminates vertices that are in conflict with more than one other vertex. Suppose that $u$ is in conflict with both $v$ and $w$. If $u$ is not deleted then both $v$ and $w$ have to be deleted. The algorithm finds the optimal solution by branching and recursively solving the instance where $u$ is deleted, and the instance where $u$ is not deleted but both $v$ and $w$ are deleted. This branching step is the bottleneck of the algorithm and gives rise to the $O(1.618^kn^{O(1)})$ and the $O(1.46^n)$ running time bounds.

The third and last phase of the algorithm deals with the case where every vertex has at most one conflict. Here we apply a divide and conquer approach that is based on our new partitioning theorem.

%

\smallskip
\noindent
{\bf Organization of the paper.} In Section \ref{pre} we set up definitions and notation, and state a few useful preliminary results.
Section~\ref{family} describes and analyzes the first phase of the algorithm.  Section~\ref{parafeed} contains the second phase, as well as the final analysis of the correctness and running time of the entire algorithm, conditioned on the correctness and running time bound of the third and last phase. In Section~\ref{par} we formally state and prove our new decomposition theorem for undirected graphs, while the description and analysis of the third phase of the algorithm is deferred to Section~\ref{dfeed}.

%% file: prelim.tex
In this paper, we work with graphs that do not contain any self loops. 
A {\em multigraph} is a graph that may contain more than one edge between the same pair of vertices. A graph is {\em mixed} if it can contain both directed and undirected edges. We will be working with mixed multigraphs; graphs that contain both directed and undirected edges, and where two vertices may have several edges between them. 

When working with a mixed multigraph $G$ we use $V(G)$ to denote the vertex set, $E(G)$ to denote the set of directed edges, and $\mathcal{E}(G)$ to denote the set of undirected edges of $G$. %
%
%
A directed edge from $u$ to $v$ is denoted by $uv$.  A \emph{supertournament} is a directed graph $T$ such that for every pair of vertices $u$, $v$ at least one (and possibly both) edges $uv$ and $vu$ are edges of $T$. Thus, every tournament is a supertorunament, but not vice versa.


\smallskip
\noindent
{\bf Graph Notation.} In a directed graph $D$, the set of \emph{out-neighbors} of a vertex $v$ is defined as $N^+(v):=\{u|vu\in E(D)\}$. Similarly, the set of \emph{in-neighbors} of a vertex $v$ is defined as $N^-(v):=\{u|uv\in E(D)\}$. A \emph{triangle} in a directed graph is a directed cycle of length $3$. Note that in this paper, whenever the term triangle is used it refers to a directed triangle. A \emph{topological sort} of a directed graph $D$ is a permutation $\pi:V(D)\mapsto [n]$ of the vertices of the graph such that for all edges $uv\in E(D)$, $\pi(u)<\pi(v)$. Such a permutation exists for a directed graph if and only if the directed graph is acyclic. For an acyclic tournament, the topological sort is unique. 

For a graph or multigraph $G$ and vertex $v$, $G - v$ denotes the graph obtained from $G$ by deleting $v$ and all edges incident to $v$. For a vertex set $S$, $G - S$ denotes the graph obtained from $G$ by deleting all vertices in $S$ and all edges incident to them.

For any set of edges $C$ directed or undirected and set of vertices $X$, the set $N_X(C)$ represents the subset of vertices of $X$ which are incident on an edge in $C$. For a vertex $v\in V(G)$, the set $N_C(v)$ represents the set of vertices $w\in V(G)$ such that there is an undirected edge $wv\in C$.

\smallskip
\noindent
{\bf Fixed Parameter Tractability.} A {\em parameterized problem} $\Pi$ is a subset of $\Sigma^* \times \mathbb{N}$. A parameterized problem $\Pi$ is said to be \emph{fixed parameter tractable}(\textsc{FPT}) if there exists an algorithm that takes as input an instance $(I, k)$ and decides whether $(I, k) \in Pi$ in time $f(k)\cdot n^c$, where $n$ is the length of the string $I$, $f(k)$ is a computable function depending only on $k$ and $c$ is a constant independent of $n$ and $k$. 

A \emph{kernel} for a parameterized problem $\Pi$ is an algorithm that given an instance $(T,k)$ runs in time polynomial in $|T|$, and outputs an instance $(T',k')$ such that $|T'|,k' \leq g(k)$ for a computable function $g$ and $(T,k) \in Pi$ if and only if $(T',k') \in \Pi$. For a comprehensive introduction to \textsc{FPT} algorithms and kernels, we refer to the book by Cygan et al.~\cite{pc_book}.

\smallskip
\noindent
{\bf Preliminary Results.}
If a tournament is acyclic then it does not contain any triangles. It is a well-known and basic fact that the converse is also true, see e.g.~\cite{Dom201076}.
\begin{lemma}\label{triangle}\cite{Dom201076}
A tournament is acyclic if and only if it contains no triangles.\end{lemma}

Lemma~\ref{triangle} immediately gives rise to a folklore greedy $3$-approximation algorithm for \textsf{TFVS}: as long as $T$ contains a triangle, delete all the vertices in this triangle.

\begin{lemma}[folklore]\label{approximation}
There is a polynomial time algorithm that given as input a tournament $T$ and integer $k$, either correctly concludes that $T$ has no feedback vertex set of size at most $k$ or outputs a feedback vertex set of size at most $3k$. 
\end{lemma}
In fact, \textsf{TFVS} has a polynomial time factor $2.5$-approximation, due to Cai et al.~\cite{CaiDZ00}. However, the simpler algorithm from Lemma~\ref{approximation} is already suitable to our needs.
%

The preliminary phase of our algorithm for \textsf{TFVS} is the kernel of Dom et al.~\cite{Dom201076}. We will need some additional properties of this kernel that we state here. Essentially, Lemma~\ref{kernel} allows us to focus on the case when the number of vertices in the input tournament is $O(k^3)$.

\begin{lemma}\label{kernel}\cite{Dom201076} There is a polynomial time algorithm that given as input a tournament $T$ and integer $k$, runs in polynomial time and outputs a tournament $T'$ and integer $k'$ such that $|V(T')| \leq |V(T)|$, $|V(T')| = O(k^3)$, $k' \leq k$, and $T'$ has a feedback vertex set of size at most $k'$ if and only if $T$ has a feedback vertex set of size at most $k$.
\end{lemma}
%



%% file: familyUndel.tex
Consider a tournament $T$ that has a feedback vertex set $H$ of size at most $k$. Then $T - H$ is acyclic. Consider now the topological order of $T - H$. Let $M$ be the set of vertices of $T-H$ whose position in the topological order is congruent to $0$ mod $\log^2 k$. We have now found a set disjoint from $H$ such that, in the topological order of $T - H$ the distance between two consecutive vertices of $M$ is $O(\log^2 k)$. We shall see later in the article that having such a set at our disposal is very useful for finding the optimum feedback vertex set $H$. Of course there is a catch; we defined $M$ using the solution $H$, but we want to use $M$ to find the solution $H$. In the rest of this section we show how to find a set $M$ with the above properties without knowing the optimum feedback vertex set $H$ in advance. We begin with a few definitions.



\begin{definition}
 Let $D$ be a directed graph. For any pair of vertices $u,v\in V(D)$ the set $\emph{between}(D,u,v)$ is defined as $N^{+}(u)\cap N^{-}(v)\setminus \{u,v\}$. 
\end{definition}
Observe that for an acyclic tournament $T$, \emph{between}$(T,u,v)$ is exactly the set of vertices coming after $u$ and before $v$ in the unique topological ordering of $T$.
\begin{definition} Let $D$ be a directed graph and $S\subseteq V(D)$. Two vertices $u,v \in S$ are called \emph{$S$-consecutive} if $uv \in E(D)$ and $\emph{between}(D,u,v)\cap S=\emptyset$.
\end{definition}
In an acyclic tournament $T$ and vertex set $S$, two vertices $u$ and $v$ in $S$ are $S$-consecutive if no other vertex of $S$ appears between $u$ and $v$ in the topological ordering. 
\begin{definition}
Let $D$ be a directed graph and $S\subseteq V(D)$. We define the set of $S$-blocks in $D$. Each pair of $S$-\emph{consecutive} vertices $u$ and $v$ defines the $S$-block $\emph{between}(D,u,v)$. Further, each vertex $u \in S$ with no in-neighbors in $S$ defines an $S$-block $N^-(u)$. Each vertex $u \in D$ with no out-neighbors in $S$ defines the $S$-block $N^+(u)$. The {\em size} of an $S$-block is its cardinality.
\end{definition}
In an acyclic tournament $T$ the $S$-blocks form a partition of $V - S$, where two vertices are in the same block if and only if no vertex of $S$ appears between them in the topological order of $T$. 
%
%
For example, consider an acyclic tournament $T=u_0u_1...u_{11}$ where vertices are topologically sorted. Let $S=\{u_i | i \mod 4=1\}$. $between(T,u_1,u_9)=\{u_2,u_3,u_4,u_5,u_6,u_7,u_8\}$. $u_5$ and $u_9$ are $S$-consecutive and $\{u_6,u_7,u_8\}$ is an $S$-block. The set of all $S$-blocks in $T$ is $\{\{u_0\},\{u_2,u_3,u_4\},\{u_6,u_7,u_8\},\{u_{10},u_{11}\}\}$.
 
 \begin{lemma}\label{undeletable} 
 There exists an algorithm that given a tournament $T$ with $|V(T)|=O(k^3)$ where $k$ is an integer, outputs a family of sets $\mathcal{M}, |\mathcal{M}| = 2^{O(\frac{k}{\log k})}$ in $2^{O(\frac{k}{\log k})}$ time, such that for every feedback vertex set $H$ of $T$ of size at most $k$, $\exists M\in \mathcal{M}$, such that :
 \begin{enumerate}
  \item $M\cap H=\emptyset$,
  \item The size of any $M$-block in $T-H$ is at most $2\log^2k$.
  \end{enumerate}
 Furthermore, $\mathcal{M}$ can be enumerated in polynomial space.
 \end{lemma}
\begin{proof}
Let $X$ be a feedback vertex set of size at most $3k$ obtained using Lemma \ref{approximation}. Let $Y:=V(T)\setminus X$ and $v_0v_1...v_{|Y|-1}$ be the topological sort of $T[Y]$ such that the edges in $T[Y]$ are directed from left to right. Color $Y$ using $\lfloor\log^2k\rfloor$ colors such that for each $c'\in [0,...,|Y|-1]$, $v_{c'}$ gets color $c'$ mod $\lfloor\log^2k\rfloor$. For each $c\in [0,...,\lfloor\log^2k\rfloor-1]$, let $Y_c$ be the set of vertices in $Y$ which get color $c$. Each $M\in\mathcal{M}$ is specified by a 4-tuple $\langle c, \hat{H}, \hat{R}, \hat{X}\rangle$ where
\begin{itemize}
\item $c$ is a color in the above coloring of $Y$,
\item $\hat{H}\subseteq Y_c$ such that $|\hat{H}|\leq \frac{k}{\log^2 k}$,
\item $\hat{R}\subseteq Y\setminus Y_c$, $|\hat{R}|\leq |\hat{H}|$, and
\item $\hat{X}\subseteq X$ such that $|\hat{X}|\leq \frac{3k}{\log^2 k}$.
\end{itemize}
For each 4-tuple $\langle c, \hat{H}, \hat{R}, \hat{X}\rangle$, let $M:=(Y_c\setminus\hat{H})\bigcup \hat{R}\bigcup \hat{X}$. Hence, $|\mathcal{M}|$ is upper bounded by the maximum number of such 4-tuples. 
\begin{eqnarray*}
|\mathcal{M}| 
&\leq& 2^{\log(\log^2k)}\times k^{\frac{6k}{\log^2k}}\times (3k)^{\frac{3k}{\log^2k}}\\
&=& 2^{O(\frac{k}{\log k})} 
\end{eqnarray*} 
Clearly, all such 4-tuples can be enumerated in polynomial space thereby providing an enumeration of $\mathcal{M}$.

We prove the correctness of the above algorithm by showing that for every feedback vertex set $H$ of $T$ of size at most $k$, $\mathcal{M}$ contains a set $M$ which satisfies the properties listed in the statement of the lemma. Let $H$ be an arbitrary feedback vertex set of $T$ of size at most $k$.


For each $j\in[0,...,\lfloor\log^2k\rfloor-1]$, let $H_j:=Y_j\cap H$. By averaging, there is a color $c$ such that $0<|H_c|\leq \frac{k}{\log^2k}$. For this color $c$, let $\hat{H}:=H_c$.
Consider a set $\hat{R}$ obtained as follows: for every vertex $v\in H_c$, pick the first vertex \emph{after} $v$ (if there is any) in $Y\setminus (Y_c \cup H)$ in the topological ordering of $T[Y]$. 
Note that $T[X\setminus H]$ is acyclic. Color $X\setminus H$ using $\lfloor\log^2k\rfloor$ colors as was done for $Y$.  Let $\hat{X}$ be the set of all vertices colored 0 in this coloring. The size of any $\hat{X}$-block in $T[X\setminus H]$ is $\log^2k$. Clearly, $|\hat{X}|\leq \frac{3k}{\log^2k}$. 

The 4-tuple $\langle c, \hat{H}, \hat{R}, \hat{X}\rangle$ described above satisfies all the properties listed in the construction of $\mathcal{M}$. Let $M:=(Y_c\setminus H_c)\bigcup \hat{R}\bigcup \hat{X}$. Clearly, $M\cap H=\emptyset$ and $M\in \mathcal{M}$. Since the size of any $[(Y_c\setminus H_c)\bigcup \hat{R}]$-\emph{block} in $Y$ is at most $\log^2k$, the size of any $M$-block in $T - H$ is at most $2\log^2k$.
\end{proof} 

Lemma~\ref{undeletable} gets us quite close to our goal. Indeed, for any feedback vertex set $H$ of size at most $k$ we will find a set $M$ such that the $M$-blocks in $T - H$ are small. However, the $M$-blocks of $T$ do not have to be small, because they could contain many vertices from $H$. The next lemma deals with this problem.

\begin{definition}
Let $D$ be a directed graph. A vertex $v\in V(D)$ is \emph{consistent} with a set $M\subseteq V(D)$ if there are no cycles in $D[M\cup v]$ containing $v$.
\end{definition}

Define a function $\mathcal{I}$ that given a directed graph $D$ and a set $M\subseteq V(D)$ outputs a set of vertices \emph{inconsistent} with $M$. Define another function $\mathcal{L}$ that given a directed graph $D$, a set $M\subseteq V(D)$ and an integer $k$ outputs a set of vertices which is the union of all $M$-blocks of size at least $2\log^4k$ in $D-\mathcal{I}(D,M)$.  
\begin{lemma}\label{blocks}
There exists an algorithm that given a tournament $T$ on $O(k^3)$ vertices where $k$ is an integer outputs a family of set pairs $\mathcal{X}=\{(M_1,P_1),(M_2,P_2),...,(M_l,P_l)\}$,$|\mathcal{X}|=2^{O(\frac{k}{\log k})}$ in $2^{O(\frac{k}{\log k})}$ time such that for every feedback vertex set $H$ of size at most $k$, there exists $(M,P)\in \mathcal{X}$ such that
 \begin{enumerate}
 \item $M\cap H=\emptyset$,
 \item $P\subseteq H$,
 \item Every vertex of $V(T)\setminus P$ is consistent with $M$, and
 \item The size of every $M$-block in $T-P$ is at most $2\log^4k$. 
 \end{enumerate}
 Furthermore, $\mathcal{X}$ can be enumerated in polynomial space.
\end{lemma}
\begin{proof} 
Use the algorithm of Lemma \ref{undeletable} to compute $\mathcal{M}$. For each $M\in \mathcal{M}$ compute the sets $\mathcal{I}(T,M)$ and $\mathcal{L}(T,M,k)$. For each $B\subseteq \mathcal{L}(T,M,k)$ such that $|B|\leq \frac{2k}{\log^2 k}$ output a pair of sets $(M, P)=(M, \mathcal{I}(T,M)\cup \mathcal{L}(T,M,k)\setminus B)$. The set $\mathcal{X}$ is the collection of all such pair of sets. 

We prove that the algorithm satisfies the stated properties. Consider a feedback vertex set $H$ of size at most $k$
%
By Lemma \ref{undeletable} there exists $M\in\mathcal{M}$ such that $M\cap H=\emptyset$. Let $C=\mathcal{I}(T,M)$ be the set of vertices that are not consistent with $M$. These vertices must belong to $H$. Since, for every vertex $v\in T-C$, $T[M\cup v]$ is an acyclic tournament, $v$ can be placed uniquely in the topological ordering of $T[M]$. Hence, for each $v\in T-C$, there is an unique $M$-block containing it. Since the size of any $M$-block in $T-H$ is at most $2\log^2 k$, the size of each $M$-block in $T-C$ will be at most $k+2\log^2k$. 

An $M$-block is called \emph{large} if its size is at least $2\log^4k$. From each \emph{large} $M$-block at least $2\log^4k-2\log^2k$ vertices belong to $H$. Hence, in total at most $\frac{k}{2\log^4k-2\log^2k}\times 2\log^2k\leq \frac{2k}{\log^2k}$ vertices from the union of \emph{large} $M$-blocks do not belong to $H$.  Since the algorithm loops over all choices of subsets $B\subseteq \mathcal{L}(T,M,k)$, $|B|\leq \frac{2k}{\log^2 k}$, $\mathcal{X}$ contains a pair $(M,P)$ satisfying the properties listed in the lemma. 

Moreover, $|\mathcal{X}|$ is bounded by the product of $|\mathcal{M}|$ and the number of subsets $B$. Now $|\mathcal{L}(T,M,k)|\leq |V(T)|$ which implies the number of subsets $B$ is at most $(k^3)^{\frac{2k}{\log^2k}}=2^{3\log k\times \frac{2k}{\log^2k}}=2^{O(\frac{k}{\log k})}$. Hence, $|\mathcal{X}|\leq 2^{O(\frac{k}{\log k})}\times 2^{O(\frac{k}{\log k})}=2^{O(\frac{k}{\log k})}$
\end{proof}

Observe that the algorithm of Lemma~\ref{blocks} does not store the family $\mathcal{X}$, but enumerates all the pairs $(M,P)\in \mathcal{X}$. Our algorithm for \textsf{TFVS} will go through all pairs in $(M,P)\in \mathcal{X}$ and for each such pair $(M, P)$ search for a feedback vertex set $H$ of size at most $k$ such that $(M, P)$ satisfy the conclusion of Lemma~\ref{blocks} for $H$. In the next section we shall see that the extra restrictions imposed on $H$ by $M$ and $P$ make it easier to find $H$.


%% file: mainAlgorithm.tex
In this section we consider the following problem. We are given as input a tournament $T$ and an integer $k$, and a pair $(M, P)$ of vertex set in $T$. The objective is to find a feedback vertex set $H$ of $T$ of size at most $k$, such that $(M, P)$ satisfy the conclusion of Lemma~\ref{blocks}.

The pair $(M,P)$ naturally leads to a partition of the vertices of $T - (P \cup M)$ into \emph{local subtournaments} corresponding to the induced graphs on the $M$-blocks in $T-P$. At this point the triangles in $T - P$ can be classified into two types: those that are entirely within a subtournament and those whose vertices are \emph{shared} between more than one subtournament. The goal of our algorithm is to eliminate all the shared triangles. When there are no such triangles left, we can solve the problem independently on each of the subtournaments. Since the subtournaments are small, even brute force search is fast enough.

To formalize our approach it is convenient to define an intermediate problem, and interpret the search for a feedback vertex set $H$ such that $(M, P)$ satisfies the conclusion of Lemma~\ref{blocks} as an instance of this intermediate problem.
Let $d$ and $t$ be two positive integers. Consider a class of mixed multigraphs $\mathcal{G}(d,t)$ in which each member is a mixed multigraph $\mathcal{T}$ with the vertex set $V(\mathcal{T})$ partitioned into vertex sets $V_1,V_2,...,V_t$ such that for each $i\in [t]$, $|V_i|\leq d$ and $T_i:=\mathcal{T}[V_i]$ is a \emph{supertournament} and the undirected edge set is $\mathcal{E}(\mathcal{T})\subseteq \bigcup_{i< j}V_i\times V_j$. \\

\noindent
\fbox{\parbox{\textwidth-\fboxsep}{
\textsc{$d$-Feedback Vertex Cover} ($d$-\textsf{FVC})\\
\textbf{Input:} A mixed multigraph $\mathcal{T}\in \mathcal{G}(d,t)$, positive integer $k$.\\
\textbf{Parameter:} $k$\\
\textbf{Task:} determine whether there exists a set $S\subseteq V(\mathcal{T})$ such that $|S|\leq k$ and $\mathcal{T}-S$ is acyclic and contains no undirected edges.
}}\\

\noindent

Now we show how \textsf{TFVS} reduces to solving $d$-feedback vertex cover problem. 
\begin{lemma}\label{fvc}
There exists a polynomial time algorithm that given a \textsf{TFVS} instance $(T,k)$ and a subset $M\subseteq V(T)$ outputs a $d$-\textsf{FVC} instance $(\mathcal{T},k)$ such that $T$ has a feedback vertex set $S$ disjoint from $M$ and $|S| \leq k$ if and only if $(\mathcal{T},k)$ is a \textsc{yes}-instance of $d$-\textsf{FVC}.
\end{lemma}
\begin{proof}
We describe an algorithm that reduces $T$ to $\mathcal{T}$ on the same set of vertices as in $T - M$. If $T[M]$ is not acyclic, then output a trivial \textsc{no}-instance. Otherwise, let $\mathcal{B}:=\{B_1,B_2,...,B_t\}$ be the set of $M$-blocks in $T$ such that the elements in $\mathcal{B}$ are indexed according to the topological order of $T[M]$. We assume that the topological order of $T[M]$ is such that the edges in $T[M]$ are directed from left to right. Let $V(\mathcal{T}):=V(T)\setminus M$. The directed edge set $E(\mathcal{T})$ is $E(T)\setminus\{e|\forall i,j\in[t], i\neq j$ and $e\in B_j\times B_i\}$. The undirected edge set in $\mathcal{T}$ is $\mathcal{E}(\mathcal{T}):=\{$undirected$(e)|i,j\in[t], i<j$ and $e\in B_j\times B_i\}$ where undirected$(e)$ is an undirected edge between the endpoints of $e$.  

Now we argue about the correctness. Since $\mathcal{T}$ is essentially a subgraph of $T-M$ with some additional undirected edges, we use the same symbol to refer to vertex or directed edge sets in both the instances. Suppose $S$ is a feedback vertex cover of $\mathcal{T}$. Clearly $S$ is disjoint from $M$. We claim that $S$ is a feedback vertex set of $T$. The triangles in $T-M$ are of two types: ones whose endpoints lie entirely in $B_i$ for some $i$ and others whose endpoints are shared among multiple $M$-blocks. Clearly, $S$ hits all the triangles within each subtournament $T[B_i]$ in $\mathcal{T}$. Hence, all that remains to show is that $S$ is also a hitting set for all triangles between different subtournaments $T[B_i]$. For the sake of contradiction suppose that there is a triangle $uvwu$ in $T-M-S$ such that not all of $u,v$ and $w$ belong to the same subtournament of $\mathcal{T}$. Then at least one edge $ab$ in this triangle is such that $a\in B_i, b\in B_j$ and $i>j$. But by the construction of $\mathcal{E}(\mathcal{T})$ there is an undirected edge between $a$ and $b$ implying that at least one of $a$ or $b$ belongs to $S$, a contradiction.  

In the other direction, suppose $S$ is a feedback vertex set of $T$ disjoint from $M$. Clearly, $S$ hits all the triangles within each subtournament $T[B_i]$ in $\mathcal{T}$. Hence, all that remains to show is that $S$ is a hitting set for $\mathcal{E}(\mathcal{T})$. Suppose not. Then there is an undirected edge $e=uv\in \mathcal{E}$ which is not hit by $S$. Consider the directed edge in $T$ corresponding to $e$. Without loss of generality, we can assume that $u\in B_i$ and $v\in B_j$ for some $i,j\in [t]$ such that the directed edge is from $u$ to $v$ and $i>j$. Now in $T$, there is a vertex $w\in M$ which lies \emph{after} all elements of $B_j$ and \emph{before} all vertices of $B_i$ and forms a triangle $vwuv$. Since, $w\notin S$, either $u\in S$ or $v\in S$, a contradiction.
\end{proof}  

In light of Lemma~\ref{fvc} we need an efficient algorithm for $d$-\textsf{FVC}. Next we will give an efficient algorithm for $d$-\textsf{FVC} and show how it can be used to obtain our claimed algorithm for \textsf{TFVS}. Our algorithm for \textsf{FVC} is based on branching on vertices that appear in at least two edges of $\mathcal{E}(\mathcal{T})$. The case when there are no such vertices has to be handled separately, the algorithm for this case is deferred to Section~\ref{dfeed}. For now, we simply state the existence of the algorithm for this case, and complete the argument using this algorithm as a black box.

\begin{lemma}\label{dfvcalgo}
There exists an algorithm running in $1.5874^s\cdot 2^{O(d^2+d\log s)}\cdot n^{O(1)}$ time which finds an optimal feedback vertex cover in a mixed multigraph $\mathcal{T}\in \mathcal{G}(d,t)$ in which the undirected edge set $\mathcal{E}(\mathcal{T})$ is disjoint and $|\mathcal{E}(\mathcal{T})|=s$. 
\end{lemma}

The proof of Lemma~\ref{dfvcalgo} can be found in Section~\ref{dfeed}. Armed with Lemma~\ref{dfvcalgo} we can give a simple and efficient algorithm for $d$-\textsf{FVC}. The algorithm is based on branching. In the course of the branching we will sometimes conclude (or guess) that a vertex $v$ is {\em not} put into the solution $S$. The operation described below encapsulates the effects of making a vertex undeletable.

In a mixed multigraph $D$, for any vertex $v$, $D\slash v$ is a mixed multigraph obtained by adding a directed edge $uw$ in $D- v$ for every $u\in N^-(v)$ and $w\in N^+(v)$. 
The next lemma shows that looking for a solution disjoint from $v$ amounts to putting all the undirected neighbors $N_{\mathcal{E}}(v)$ of $v$ into the solution, and finding the optimum solution of $(\mathcal{T}- N_{\mathcal{E}}(v))\slash v$.
%
%
%
%
\begin{lemma}\label{reduce}
Let $(\mathcal{T},k)$ be a $d$-\textsf{FVC} instance. If for any vertex $v\in V(\mathcal{T})$ such that $N_{\mathcal{E}}(v)=\emptyset$, then $(\mathcal{T},k)$ has a solution of size at most $k$ not containing $v$ if and only if $(\mathcal{T}\slash v,k)$ is a \textsc{yes}-instance.
\end{lemma}
\begin{proof}
Let $S$ be a feedback vertex cover of $\mathcal{T}$ of size at most $k$ not containing $v$. We show that $S$ is a feedback vertex cover of $\mathcal{T}\slash v$. Clearly, $S$ hits every undirected edge in $\mathcal{T}\slash v$. For the sake of contradiction suppose there is a cycle of length 3 containing $v$ in $\mathcal{T}\slash v$ not hit by $S$. If this cycle is in $\mathcal{T}-v$, then it is hit by $S$. Hence, the triangle must contain an edge not in $\mathcal{T}-v$. Note that $\mathcal{T}\slash v$ is obtained by adding a directed edge $yx$ for every triangle $xyvx$ in $\mathcal{T}-v$ thereby creating a 2-cycle between $x$ and $y$. Since $v\notin S$, either $x\in S$ or $y\in S$, a contradiction. For the same reason there are no cycles of length 2 in $\mathcal{T}\slash v-S$. 

Now suppose $S$ is a feedback vertex cover of $\mathcal{T}\slash v$. Since every cycle in $\mathcal{T}-v$ is a cycle in $\mathcal{T}\slash v$ which are hit by $S$, we need to consider cycles in $\mathcal{T}$ containing $v$. But, for every such cycle $xyvx$ we have a cycle $xyx$ of length 2 in $\mathcal{T}\slash v$ which is hit by $S$, we have that $\mathcal{T}-S$ is acyclic. 
\end{proof}

\begin{lemma}\label{fvcalgo}
There exists an algorithm for $d$-\textsf{FVC} running in $1.466^n\cdot 2^{O(d^2+d\log n)}$ time and in $1.618^k\cdot 2^{O(d^2+d\log k)}\cdot n^{O(1)}$ time.
\end{lemma}

\begin{proof} We describe a recursive algorithm which searches for a potential solution $S$ of size at most $k$ by branching. For any vertex $v$, let $N_{\mathcal{E}}(v)$ denote the set of vertices $w$ such that $vw\in \mathcal{E}$. Let $s=|\mathcal{E}|$. As long as there is a vertex $v$ such that $|N_{\mathcal{E}}(v)|\geq 2$ and $k>0$, the algorithm branches by considering both the possibilities: either $v\in S$ or $v\notin S$. In the branch in which $v$ is picked, $n$ and $k$ are decreased by 1 each and $v$ is removed from the graph. In the other branch, $N_{\mathcal{E}}(v)$ is added to $S$, and $k$ is decreased by $|N_{\mathcal{E}}(v)|$. 
At the same time, $N_{\mathcal{E}}(v)$ is removed from the graph. Since, $N_{\mathcal{E}}(v)=\emptyset$, by Lemma \ref{reduce} $(\mathcal{T},k)$ is reduced to $(\mathcal{T}\slash v,k)$. Thus the number of vertices is decreased by $|N_{\mathcal{E}}[v]|$. The algorithm stops branching further in a branch in which either $k<0$ or $k>0$ and for every vertex $v$, $|N_{\mathcal{E}}(v)|\leq 1$. In the case that $k<0$, the algorithm terminates the branch and moves on to other branches. In the other case, if $|\mathcal{E}(T)|>k$, the algorithm terminates that branch, otherwise the algorithm of Lemma \ref{dfvcalgo} is applied. If  the size of the optimal solution of Lemma \ref{dfvcalgo} is at most $k$, then the algorithm outputs \textsc{yes} and terminates, otherwise the algorithm moves to another branch. If the algorithm fails to find any solution of size at most $k$ in every branch, it outputs \textsc{no}. 

Now we do the runtime analysis of the algorithm. At each internal node of the recursion tree, the algorithm spends polynomial time. At a leaf node, either the algorithm terminates or makes a call to the algorithm of Lemma \ref{dfvcalgo} with parameter $s=|\mathcal{E}(\mathcal{T})|$ which is at most $k$. So, we need to bound the number of times Lemma \ref{dfvcalgo} is called with parameter $s$ for each value of $s$ in $[k]$. Note that for any $s$, the smallest value of $k$ with which a call to the algorithm of Lemma \ref{dfvcalgo} is made is $s$. Therefore, for each value of $s\in [k]$, the number of calls to the algorithm of Lemma \ref{dfvcalgo} is bounded by the number of nodes in the recursion tree with $k=s$. The recurrence relation for bounding the number of leaves in the recursion tree of the algorithm is given by:
\begin{equation*}
f_s(k)\leq f_s(k-1)+f_s(k-2)
\end{equation*}
which solves to $f_s(k)\leq 1.618^{k-s}$ as $f_s(k)\leq 1$ for $k=s$. Hence, the runtime of the algorithm is upper bounded by $\sum\limits_{s=1}^{k}1.618^{k-s}\times 1.5874^s\cdot 2^{O(d^2+d\log s)}\cdot n^{O(1)}\leq 1.618^k\cdot 2^{O(d^2+d\log k)}\cdot n^{O(1)}$.

We can do a similar analysis to bound the runtime in terms of $n$. Note that in direct correspondence with the fact that when ever $k$ decreases by 1, $n$ decreases by 1 and whenever $k$ decreases by $x\geq 2$, $n$ decreases by $x+1$, we get the following recurrence relation:
\begin{equation*}
f_s(n)\leq f_s(n-1)+f_s(n-3)
\end{equation*}
implying $f_s(n)\leq 1.466^{n-s}$ as $f_s(k)\leq 1$ for $n=s$. If $s$ is the size of of the graph, then the largest value of $|\mathcal{E}|$ with which a call to the algorithm of Lemma \ref{dfvcalgo} is made, is at most $\frac{s}{2}$. Hence, the runtime of the algorithm is upper bounded by $\sum\limits_{s=1}^n 1.466^{n-s}\times 1.5874^{\frac{s}{2}}\cdot 2^{O(d^2+d\log n)}\cdot n^{O(1)}\leq 1.466^n\cdot 2^{O(d^2+d\log n)}\cdot n^{O(1)}$. 
\end{proof}

Having shown an efficient algorithm for $d$-\textsf{FVC}, we are now in position to prove our main theorem.
\begin{theorem}\label{pfvs}
There exists an algorithm for \textsf{TFVS} running in $O(1.466^n)$ time and in $O(1.618^k + n^{O(1)})$ time.
\end{theorem}
\begin{proof}
The algorithm begins by running the kernelization algorithm of Lemma \ref{kernel} for the given \textsf{TFVS} instance. In the remainder we assume that $n=O(k^3)$. Next the algorithm proceeds to apply Lemma \ref{blocks} to create a family of set pairs $\mathcal{X}$. 
For each set pair $(M,P)\in \mathcal{X}$ it determines whether there is a feedback vertex set $H$ of size at most $k$ such that $H\cap M=\emptyset$ and $P\subseteq H$ as follows:

First it runs the algorithm of Lemma \ref{fvc} with input $(T-P,k-|P|)$ and $M$ to reduce the problem to an equivalent $d$-\textsf{FVC} instance which is then passed to the algorithm of Lemma \ref{fvcalgo} as input. The algorithm outputs \textsc{yes} and terminates if the output of the algorithm of Lemma \ref{fvcalgo} is \textsc{yes}.  If no solution of size at most $k$ is obtained for any set pair in $\mathcal{X}$, the algorithm outputs \textsc{no} and terminates.

The correctness of the algorithm follows from Lemma \ref{blocks} and Lemma \ref{fvc}. The running time of the algorithm is upper bounded by $|\mathcal{X}|$ times (the runtime of the algorithm of Lemma \ref{fvcalgo}). Since $|\mathcal{X}|=2^{o(k)}$ and the bulk of the algorithm is run on a tournament with at most $O(k^3)$ vertices the total time used by the algorithm is upper bounded by $O(1.466^n)$ and $O(1.618^k +  n^{O(1)})$.
\end{proof}

We have now proved our main result, assuming the correctness of Lemma~\ref{dfvcalgo}. The remainder of the paper is devoted to proving Lemma~\ref{dfvcalgo}. The engine of the algorithm of Lemma~\ref{dfvcalgo} is a new graph partitioning theorem. The next section contains the statement and proof of this theorem, while Section~\ref{dfeed} wraps up the proof of Lemma~\ref{dfvcalgo}, thereby completing the proof of Theorem~\ref{pfvs}.

%% file: partition.tex
Given an undirected graph $G,|E(G)|=m$, if each vertex in $V(G)$ is colored \emph{red} or \emph{blue} uniformly at random, then in expectation there will be $\frac{m}{4}$ \emph{red} edges and $\frac{m}{4}$ \emph{blue} edges, where a red edge is an edge whose both endpoints are red and a blue edge is an edge whose both endpoints are blue. Using Chebysev inequality it can be shown that, with high probability, the number of red or blue edges will be within $O(\sqrt{md})$ of $\frac{m}{4}$ where $d$ is the maximum degree of a vertex in the graph. A proof of this fact is skipped in favor of a local search algorithm which runs in polynomial time and provides a coloring with smaller deviation from the expected value than random coloring.
\begin{theorem}\label{partition}
 Given an undirected, multigraph without self-loops and isolated vertices $G$ of maximum degree at most $d$ and $|E(G)|=m$, there exists a partition $(A,B)$ of $V(G)$ such that 
 \begin{itemize}
 \item $\frac{m}{4}- \frac{d}{2}\leq|E(G[A])|\leq\frac{m}{4}+ \frac{d}{2}$, 
 \item $\frac{m}{4}- \frac{d}{2}\leq|E(G[B])|\leq\frac{m}{4}+ \frac{d}{2}$, and 
 \item $\frac{m}{2}- d\leq|E(G[A,B])|\leq\frac{m}{2}+ d$
\end{itemize}
 where $E(G[A,B])$ is the set of edges with one endpoint in $A$ and other in $B$. Furthermore, there is a polynomial time algorithm to obtain this partition.
\end{theorem}
\begin{proof}
The following local search algorithm is used to obtain the desired partition:

At each step, the algorithm maintains a partition $(A,B)$ of $V(G)$. As long as there exists a vertex $v\in A$ (or $v\in B$) such that moving it to other part decreases the measure $\mu=||E(G[A])|-\frac{m}{4}|+||E(G[B])|-\frac{m}{4}|$, the algorithm changes the partition to $(A\setminus v,B\cup v)$ (or $(A\cup v,B\setminus v)$). The algorithm terminates if no vertex can be moved. Since $\mu\leq m$ and in each step, it decreases by at least one, above algorithm terminates in polynomial time.\\
\emph{Correctness}: Let $m_A:=|E(G[A])|,m_B:=|E(G[B])|$, and $m_C:=|E(G[A,B])|$. Let $x:=m_A-\frac{m}{4}$ and $y:=m_B-\frac{m}{4}$ when the algorithm terminates. Then, $\mu=|x|+|y|$. For any vertex $v$, let $a_v$ denote the number of edges incident on $v$ whose other endpoints are in $A$ and $b_v$ denote the number of edges incident on $v$ whose other endpoints are in $B$. Clearly, for every vertex $v$, $a_v+b_v\leq d$. Suppose that a vertex $v\in A$ is moved to $B$. The new partition is $(A',B')=(A\setminus v,B\cup v)$. Then, $m_{A'}=\frac{m}{4}+x-a_v, m_{B'}=\frac{m}{4}+y+b_v$ and the measure at this partition is $\mu'=|x-a_v|+|y+b_v|$. Define $\delta_A^v:=\mu'-\mu=|x-a_v|-|x|+|y+b_v|-|y|$. Similarly, if a vertex $v\in B$ moves to $A$ creating new partition $(A',B')=(A\cup v,B\setminus v)$, we can define $\delta_B^v:=\mu'-\mu=|x+a_v|-|x|+|y-b_v|-|y|$. Note that since the algorithm has terminated, for any vertex $v\in V(G)$, $\delta_A^v\geq0$ and $\delta_B^v\geq0$. Then, the claim of the theorem is that $|x|\leq \frac{d}{2}$ and $|y|\leq \frac{d}{2}$. For the sake of contradiction assume the following possible values of $x$ and $y$: 
\begin{description}
    \item[$x>\frac{d}{2},y>\frac{d}{2}$] : 
    Consider moving a vertex $v\in A$ to $B$. Then, $\delta_A^v=|x-a_v|-|x|+b_v$. Suppose that $x<a_v,\delta_A^v=a_v-x-x+b_v=a_v+b_v-2x$. But, for every vertex $v\in V$, $a_v+b_v\leq d$ which implies $\delta_A^v<0$, a contradiction. 
    
    Hence, for all vertices $v\in A$, $x\geq a_v, \delta_A^v=x-a_v-x+b_v=b_v-a_v$. If $v\in A$ is such that $a_v>b_v$, then $\delta_A^v<0$, a contradiction. Hence, for all vertices $v\in A, a_v\leq b_v$. Then, $\sum\limits\limits_{v\in A} a_v\leq \sum\limits_{v\in A} b_v\implies 2m_A\leq m_C\implies m_C\geq \frac{m}{2}+2x>\frac{m}{2}+d$ which is a contradiction. 
    
    \item[$x<-\frac{d}{2},y <-\frac{d}{2}$] : Consider moving a vertex $v\in A$ to $B$. Then, $\delta_A^v=|y+b_v|-|y|+a_v$. Suppose that $|y|<b_v,\delta_A^v=b_v-|y|-|y|+a_v=a_v+b_v-2|y|$. But, for all vertices $v$, $a_v+b_v\leq d$ which implies that $a_v-2|y|+b_v < 0$, i.e. $\delta_A^v<0$, a contradiction. 
    
    Hence, for every vertex $v\in A,$ we have that $|y|\geq b_v$ and therefore, $\delta_A^v=|y|-b_v-|y|+a_v=a_v-b_v$. If $v\in A$ is such that $a_v<b_v$, then  $\delta_A^v<0$, a contradiction. Hence, for all vertices $v\in A, a_v\geq b_v$. This implies that $\sum\limits_{v\in A} a_v\geq \sum\limits_{v\in A} b_v\implies 2m_A\geq m_C\implies m_C\leq \frac{m}{2}+2x<\frac{m}{2}-d$ which is a contradiction.

    \item[$x>\frac{d}{2}, y<-\frac{d}{2}$] : Consider moving a vertex $v\in A$ to $B$. Then, $\delta_A^v:=\mu'-\mu=|x-a_v|-|x|+|y+b_v|-|y|<0$ as $|x-a_v|-|x|\leq0$ and $|y+b_v|-|y|\leq 0$ and at least one of the inequalities is strict, hence a contradiction.   
    
    \item[$y>\frac{d}{2}, x<-\frac{d}{2}$] : Similar to the previous case.
    \item[$x>\frac{d}{2},|y|\leq \frac{d}{2}$] : 
    Consider moving a vertex $v\in A$ to $B$. Suppose that $x<a_v$, then $\delta_A^v=a_v-2x+|y+b_v|-|y|\leq a_v-2x+b_v<0$, a contradiction. 
    
     Hence, for every vertex $v\in A$, $x\geq a_v$, then $\delta_A^v=|y+b_v|-|y|-a_v\leq b_v-a_v$. If $v\in A$ is such that $a_v>b_v$, then $\delta_A^v<0$, a contradiction. Hence, for every vertex $v\in A$, we have that $a_v\leq b_v$. This implies that $\sum\limits\limits_{v\in A} a_v\leq \sum\limits_{v\in A} b_v\implies 2m_A\leq m_C\implies m_C\geq \frac{m}{2}+2x>\frac{m}{2}+d$ which is a contradiction. 
   \item[$y>\frac{d}{2},|x|\leq \frac{d}{2}$] :  Similar to the previous case.  
  \item[$x<-\frac{d}{2},|y|\leq \frac{d}{2}$] : Consider moving a vertex $v\in B$ to $A$.  If $|x|\geq a_v$, then $\delta_B^v=a_v-2|x|+|y-b_v|-|y|\leq a_v-2|x|+b_v<0$, a contradiction. So, for every vertex $v\in B$, $|x|<a_v$ and $0\leq\delta_B^v=|x|-a_v-|x|+|y-b_v|-|y|\leq -a_v+b_v$. Hence, for each vertex $v\in B$, $a_v\leq b_v$. This implies $\sum\limits\limits_{v\in B} a_v\leq \sum\limits_{v\in B} b_v\implies m_C\leq 2m_B\implies m_C\leq \frac{m}{2}+2y<\frac{m}{2}$ which is a contradiction. 
  \item[$y<-\frac{d}{2},|x|\leq \frac{d}{2}$] :  Similar to the previous case. 
  \end{description}
Hence, $|x|\leq \frac{d}{2}$ and $|y|\leq \frac{d}{2}$. This implies that $\frac{s}{2}- d\leq m_C\leq\frac{s}{2}+ d$. This concludes the proof of the theorem. 
\end{proof}

%% file: fvcDivideConquer.tex
Now that we are equipped with Theorem \ref{partition}, we are almost ready to prove Lemma~\ref{dfvcalgo}. First we show a lemma that encapsulates the use of Theorem~\ref{partition} inside the algorithm of Lemma~\ref{dfvcalgo}.
\begin{lemma}\label{fvcpart}
There exists a polynomial time algorithm that given a mixed multigraph $\mathcal{T}\in \mathcal{G}(d,t)$ with disjoint undirected edge set $\mathcal{E}(\mathcal{T})$ outputs a partition $(X,Y)$ of $V(\mathcal{T})$ such that there are no directed edge with one endpoint in $X$ and other in $Y$ and
\begin{itemize}
\item $||\mathcal{E}(X)\cap \mathcal{E}|-\frac{s}{4}|\leq \frac{d}{2}$,
\item $||\mathcal{E}(Y)\cap \mathcal{E}|-\frac{s}{4}|\leq \frac{d}{2}$ and
\item $||\mathcal{E}(X,Y)\cap \mathcal{E}|-\frac{s}{2}|\leq d$
\end{itemize}
where  $s=|\mathcal{E}(\mathcal{T})|$ and $\mathcal{E}(X)$ is the set of undirected edges in $\mathcal{T}[X]$, $\mathcal{E}(Y)$ is the set of undirected edges in $\mathcal{T}[Y]$ and $\mathcal{E}(X,Y)$ is the set of undirected edges with one endpoint in $X$ and other in $Y$.
\end{lemma}
\begin{proof}
Construct an undirected, multigraph $Z$ such that $V(Z)=\{z_i|i\in [t]\}$ and $E(Z)=\{z_iz_j|uv\in \mathcal{E}, u\in V_i, v\in V_j\}$. Run the algorithm of Theorem \ref{partition} to get the partition $(A,B)$ of $V(Z)$. Output $X:=\bigcup\limits_{i,z_i\in A}V_i$ and $Y:=\bigcup\limits_{i,z_i\in B}V_i$. 

Since $\mathcal{E}$ is disjoint and for each $i\in [t]$, $|V_i|\leq d$, maximum degree of a vertex in $Z$ is at most $d$. Hence, the correctness of the algorithm and the size bound in the lemma follows from Theorem \ref{partition}. 
\end{proof}
We are now ready to prove Lemma~\ref{dfvcalgo}. For convenience we re-state it here.

\smallskip
\noindent
{\bf Lemma~\ref{dfvcalgo}} {\em There exists an algorithm running in $1.5874^s\cdot 2^{O(d^2+d\log s)}\cdot n^{O(1)}$ time which finds an optimal feedback vertex cover in  a mixed multigraph $\mathcal{T}\in \mathcal{G}(d,t)$ in which the undirected edge set $\mathcal{E}(\mathcal{T})$ is disjoint and $|\mathcal{E}(\mathcal{T})|=s$.} 

\begin{proof}
The algorithm maintains a set $S$ which is initialized to the empty set $\emptyset$. If the underlying undirected graph of $\mathcal{T}$ is disconnected, then the algorithm solves each connected component independently and outputs $S$ as the union of sets returned for each component. If $s\leq d$, then $S$ is an optimal solution set obtained by a brute force search in the instance. If $s> d$, the algorithm obtains a partition $(X,Y)$ of $V(\mathcal{T})$ by running the algorithm of Lemma \ref{fvcpart}. Then, it loops over all subsets $C\subseteq \mathcal{E}(X,Y)$, calling itself recursively on $\mathcal{T}\bigl[V(\mathcal{T})\setminus(V_X(C)\cup V_Y(\mathcal{E}(X,Y)\setminus C))]$ and computes $S_C:=V_X(C)\cup V_Y(\mathcal{E}(X,Y)\setminus C)\cup S'$ where $S'$ is the set returned at the recursive call. 
Finally, the algorithm outputs the smallest set $S_C$ over all choices of $C\subseteq \mathcal{E}(X,Y)$.

Now to argue about the correctness of the algorithm, we use induction on $|\mathcal{E}(\mathcal{T})|$. In the base case $|\mathcal{E}(\mathcal{T})|\leq d$, $S$ is an optimal feedback vertex cover. As the induction hypothesis, suppose that the algorithm outputs an optimal solution for $d<|\mathcal{E}(\mathcal{T})|<s$. Consider $|\mathcal{E}(\mathcal{T})|=s$. Note that for any $C\subseteq \mathcal{E}(X,Y)$, $S_C$ is a $d$-feedback vertex cover as $V_X(C)\cup V_Y(\mathcal{E}(X,Y)\setminus C)$ is a hitting set for $\mathcal{E}(X,Y)$ and by the induction hypothesis, $S'$ is an optimal solution for $\mathcal{T}[V(\mathcal{T})\setminus(V_X(C)\cup V_Y(\mathcal{E}(X,Y)\setminus C))]$. At the same time, for any $C\subseteq \mathcal{E}(X,Y)$, $|V_X(C)\cup V_Y(\mathcal{E}(X,Y)\setminus C)|=|\mathcal{E}(X,Y)|$ which is the size of the smallest hitting set for $\mathcal{E}(X,Y)$. Let $S_o$ be an optimal solution and $C_o:=\mathcal{E}(S_o\cap X,Y\setminus S_o)$. Then, we claim that $|S_{C_o}|=|S_o|$. Clearly, $|S_{C_o}|\geq|S_o|$. Now, $S_o\setminus V_X(C_o)\cup V_Y(\mathcal{E}(X,Y)\setminus C_o)$ is a $d$-feedback vertex cover for $\mathcal{T}[V(\mathcal{T})\setminus (V_X(C_o)\cup V_Y(\mathcal{E}(X,Y)\setminus C_o))]$. Therefore, $|S'|\leq |S_o\setminus (V_X(C_o)\cup V_Y(\mathcal{E}(X,Y)\setminus C_o)|=|S_o|-|\mathcal{E}(X,Y)|\implies |S_{C_o}|\leq|S_o|$, thus proving the claim. 

Now we proceed to the runtime analysis of the algorithm. Let $h(s,d)$ be the maximum number of leaves in the recursion tree of the algorithm when run on an input with parameters $s$ and $d$. Since, in each recursive call, $s$ decreases by at least 1, the depth of the recursion tree is at most $s$. In each internal node of the recursion tree, the algorithm spends polynomial time in size of the input and in each leaf, it spends at most $2^{O(d^2)}$ time as the total number of vertices in each connected component of $\mathcal{T}$ is $O(d^2)$. Thus, the runtime of the algorithm on any input with parameters $s$ and $d$ is upper bounded by $h(s,d)\times 2^{O(d^2)}\times n^{O(1)}$. To upper bound $h(s,d)$, first note that $h(a,d)+h(b,d)\leq h(a+b,d)$ because $h(a,d)$ and $h(b,d)$ represent the number of leaves of two independent subtrees. Now for each $C\subseteq \mathcal{E}(X,Y)$, in $\mathcal{T}[V(\mathcal{T})\setminus (V_X(C)\cup V_Y(\mathcal{E}(X,Y)\setminus C))]$, the undirected edge set $\mathcal{E}(X,Y)=\emptyset$. Hence, the algorithm effectively solves $\mathcal{T}[V(\mathcal{T})\setminus (V_X(C)]$ and $\mathcal{T}[V(\mathcal{T})\setminus V_Y(\mathcal{E}(X,Y)\setminus C)]$ independently where by Lemma \ref{fvcpart}, the number of undirected edges is at most $\frac{s}{4}+\frac{d}{2}$ for each instance. Again by Lemma \ref{fvcpart}, $|\mathcal{E}(X,Y)|\leq \frac{s}{2}+d$. Hence, the number of choices for $C\subseteq \mathcal{E}(X,Y)$ is at most $2^{\frac{s}{2}+d}$. As we have seen for each $C$, the algorithm calls itself twice on graphs with the undirected edge set size at most $\frac{s}{4}+\frac{d}{2}$. So in total, the algorithm makes $2^{\frac{s}{2}+d+1}$ recursive calls with \emph{parameter} $\frac{s}{4}+\frac{d}{2}$. Thus $h(s,d)$ is upper bounded by the recurrence relation $h(s,d)\leq 2^{1+\frac{s}{2}+d} h(\frac{s}{4}+\frac{d}{2},d)$
%
%
%
%
which solves to $h(s,d)=1.5874^s\cdot 2^{O(d\log s)}$. Hence, the runtime of the algorithm is bounded by $1.5874^s\cdot 2^{O(d\log s)}\times 2^{O(d^2)}\times n^{O(1)}=1.5874^s\cdot 2^{O(d^2+d\log s)}\cdot n^{O(1)}$.
\end{proof}

The proof of Lemma~\ref{dfvcalgo} completes the proof of our main result, an algorithm for \textsf{TFVS} with running time upper bounded by $O(1.466^n)$ and by $O(1.618^k + n^{O(1)})$.

%% file: fvst.bbl
\begin{thebibliography}{10}

\bibitem{BafnaBF99}
Vineet Bafna, Piotr Berman, and Toshihiro Fujito.
\newblock A 2-approximation algorithm for the undirected feedback vertex set
  problem.
\newblock {\em {SIAM} J. Discrete Math.}, 12(3):289--297, 1999.

\bibitem{CaiDZ00}
Mao{-}cheng Cai, Xiaotie Deng, and Wenan Zang.
\newblock An approximation algorithm for feedback vertex sets in tournaments.
\newblock {\em {SIAM} J. Comput.}, 30(6):1993--2007, 2000.

\bibitem{ChenLLOR08}
Jianer Chen, Yang Liu, Songjian Lu, Barry O'Sullivan, and Igor Razgon.
\newblock A fixed-parameter algorithm for the directed feedback vertex set
  problem.
\newblock {\em J. {ACM}}, 55(5), 2008.

\bibitem{pc_book}
Marek Cygan, Fedor~V. Fomin, Lukasz Kowalik, Daniel Lokshtanov, D{\'{a}}niel
  Marx, Marcin Pilipczuk, Michal Pilipczuk, and Saket Saurabh.
\newblock {\em Parameterized Algorithms}.
\newblock Springer, 2015.

\bibitem{CyganNPPRW11}
Marek Cygan, Jesper Nederlof, Marcin Pilipczuk, Michal Pilipczuk, Johan M.~M.
  van Rooij, and Jakub~Onufry Wojtaszczyk.
\newblock Solving connectivity problems parameterized by treewidth in single
  exponential time.
\newblock In {\em {IEEE} 52nd Annual Symposium on Foundations of Computer
  Science, {FOCS} 2011, Palm Springs, CA, USA, October 22-25, 2011}, pages
  150--159, 2011.

\bibitem{DomGHNT06conf}
Michael Dom, Jiong Guo, Falk H{\"{u}}ffner, Rolf Niedermeier, and Anke
  Tru{\ss}.
\newblock Fixed-parameter tractability results for feedback set problems in
  tournaments.
\newblock In {\em Algorithms and Complexity, 6th Italian Conference, {CIAC}
  2006, Rome, Italy, May 29-31, 2006, Proceedings}, pages 320--331, 2006.

\bibitem{Dom201076}
Michael Dom, Jiong Guo, Falk H{\"{u}}ffner, Rolf Niedermeier, and Anke
  Tru{\ss}.
\newblock Fixed-parameter tractability results for feedback set problems in
  tournaments.
\newblock {\em J. Discrete Algorithms}, 8(1):76--86, 2010.

\bibitem{erdHos1965independent}
P~Erd{\H{o}}s and L~P{\'o}sa.
\newblock On independent circuits contained in a graph.
\newblock {\em Canad. J. Math}, 17:347--352, 1965.

\bibitem{EvenNSS98}
Guy Even, Joseph Naor, Baruch Schieber, and Madhu Sudan.
\newblock Approximating minimum feedback sets and multicuts in directed graphs.
\newblock {\em Algorithmica}, 20(2):151--174, 1998.

\bibitem{FominK10}
Fedor~V. Fomin and Dieter Kratsch.
\newblock {\em Exact Exponential Algorithms}.
\newblock Texts in Theoretical Computer Science. An {EATCS} Series. Springer,
  2010.

\bibitem{GJ79}
Michael~R. Garey and David~S. Johnson.
\newblock {\em Computers and Intractability: A Guide to the Theory of
  NP-Completeness}.
\newblock Series of Books in the Mathematical Sciences. W. H. Freeman and Co.,
  1979.

\bibitem{GaspersM13}
Serge Gaspers and Matthias Mnich.
\newblock Feedback vertex sets in tournaments.
\newblock {\em Journal of Graph Theory}, 72(1):72--89, 2013.

\bibitem{KociumakaP14}
Tomasz Kociumaka and Marcin Pilipczuk.
\newblock Faster deterministic feedback vertex set.
\newblock {\em Inf. Process. Lett.}, 114(10):556--560, 2014.

\bibitem{RamanS06}
Venkatesh Raman and Saket Saurabh.
\newblock Parameterized algorithms for feedback set problems and their duals in
  tournaments.
\newblock {\em Theor. Comput. Sci.}, 351(3):446--458, 2006.

\bibitem{Razgon07}
Igor Razgon.
\newblock Computing minimum directed feedback vertex set in
  o(1.9977\({}^{\mbox{n}}\)).
\newblock In {\em Theoretical Computer Science, 10th Italian Conference,
  {ICTCS} 2007, Rome, Italy, October 3-5, 2007, Proceedings}, pages 70--81,
  2007.

\bibitem{reed1996packing}
Bruce Reed, Neil Robertson, Paul Seymour, and Robin Thomas.
\newblock Packing directed circuits.
\newblock {\em Combinatorica}, 16(4):535--554, 1996.

\bibitem{XiaoN15}
Mingyu Xiao and Hiroshi Nagamochi.
\newblock An improved exact algorithm for undirected feedback vertex set.
\newblock {\em J. Comb. Optim.}, 30(2):214--241, 2015.

\end{thebibliography}
